\documentclass{amsart}
\date{March 15, 2026}
\usepackage{amsthm,amssymb,bbm}
\usepackage{mathtools}

\newtheorem{theorem}{Theorem}

\newtheorem{lemma}{Lemma}

\def\tr{{\rm tr \,}}

\def\cz{\mathbb{C}}
\def\nz{\mathbb{N}}
\def\rz{\mathbb{R}}

\def\bp{{\bold p}}

\def\bx{{\bold x}}

\def\by{{\bold y}}
\def\bz{\mathbf{z}}

\def\gQ{\mathfrak{Q}}
\def\gS{\mathfrak{S}}

\def\rd{\mathrm{d}}
\def\ri{\mathrm{i}}

\def\atp{\bigwedge\nolimits} 

\def\ca{\mathcal{E}^\mathrm{CA}}
\def\hf{\mathcal{E}^\mathrm{HF}}
\def\m{\mathcal{E}^\mathrm{M}}

\title[Cs\'anyi and Arias]{On Cs\'anyi's and Arias' Functional for
  Ground States Energy of Multi-Particle Fermion Systems: Asymptotics}
\author[H. Siedentop]{Heinz Siedentop}
\address{Mathematisches  Institut\\
  Ludwig-Maximilians-Universit\"at M\"unchen\\
  Theresienstr. 39\\ 80333 M\"unchen\\Germany\\ and Munich Center for
  Quantum Science and Technology (MCQST)\\ Schellingstr. 4\\ 80799
  M\"unchen\\ Germany} \email{h.s@lmu.de}



\thanks{This work was partially supported by the Deutsche
  Forschungsgemeinschaft via the TRR 352 – Project-ID 470903074 and
  EXC-2111-390814868}

\begin{document}

\begin{abstract}
  We show that Cs\'anyi's and Arias' energy functional of the
  one-particle reduced density matrix is bounded from below by the
  M\"uller functional and bounded from above by the Hartree-Fock
  functional. We use this fact to derive an asymptotic expansion of
  the ground state energy of this functional which agrees with the
  quantum energy to third order.
\end{abstract}
\maketitle
\section{Introduction}

Cs\'anyi and Arias \cite{CsanyiArias2000} introduced an energy
functional, hereafter referred to as the CA functional of the
one-particle reduced density matrix $\gamma$ which they
call``corrected Hartree-Fock functional''. To define it we first
define the Hartree-Fock functional
\begin{multline}
  \label{eq:hf}
  \hf(\gamma)\\ :=\tr \left((\tfrac12\bp^2-V)\gamma\right) +
    \underbrace{\tfrac12\int_{\rz^3}\rd \bx \int_{\rz^3}
      \rd\by{\rho_\gamma(\bx)\rho_\gamma(\by)\over|\bx-\by|}}_{=:
      D[\rho_\gamma]}- \underbrace{\tfrac12 \int_\Gamma\rd x
      \int_\Gamma\rd y {|\gamma(x,y)|^2\over|\bx-\by|}}_{=:X[\gamma]}
\end{multline}
where $\rho_\gamma$ is the particle density of $\gamma$, i.e.,
$\rho_\gamma(\bx) := \sum_n\sum_{\sigma=1}^q\lambda_n|\xi_n(x)|^2$
with the notation $x=(\bx,\sigma)\in \Gamma:=\rz^3\times \cz^q$ and
$\int_\Gamma\rd x:= \int_{\rz^3} \rd\bx \sum_{\sigma=1}^q$. Here
$\lambda_1,\lambda_2,...$ and $\xi_1,\xi_2,...$ are the eigenvalues
and eigenfunctions of $\gamma$, i.e.,
$\gamma=\sum_n \lambda_n|\xi_n\rangle\langle\xi_n|$. Moreover, $q$ is
the number of spin states per particle, i.e., for non-relativistic
electrons $q=2$. Eventually, $T(\bp)$ is the kinetic energy written in
terms of the momentum $\bp:=-\ri\nabla$ and $V$ is the external
potential. As we will see in the proof, the exact form of the kinetic
energy $T$ and of the potential $V$ is not relevant for our main
result. However, for definiteness we assume $T(\bp)=\bp^2/2$ and
$V(\bx)=Z/|\bx|$ which is essential for the asymptotic expansion
\eqref{Korrolar-ca}.

Using this notation the
CA functional is
\begin{equation}
  \label{eq:ca}
  \ca(\gamma):= \hf(\gamma) -X[\sqrt{\gamma(1-\gamma)}].
\end{equation}

Cs\'anyi and Arias found the functional $\ca$ by approximating the
two-particle density matrix quadratically keeping some quantum
statistical properties, some symmetries, and normalization. The only
other functional that they found under those requirements is the
M\"uller functional (in CA's diction the ``corrected Hartree
functional''): 
\begin{equation}
  \label{eq:muller}
  \m(\gamma):=\tr ((T(\bp)-V(\bx))\gamma) + D[\rho_\gamma]-X[\sqrt\gamma].
\end{equation}

We consider these functionals on certain subsets of the set
$\gS^1(L^2(\Gamma:\cz))$ of trace class operators on square integrable
complex valued functions on $\Gamma$, namely
\begin{align}
  \gQ:=&\{\gamma\in \gS^1(L^2(\Gamma:\cz))\big |0\leq\gamma\leq1,\ T(\bp) \gamma \in \gS^1(L^2(\Gamma:\cz)\},\\
  \gQ_N:=&\{\gamma\in\gQ \big| \tr\gamma\leq N\}
\end{align}
where $N\in\nz$ is the number of electrons, i.e., the set of 1-pdms
with finite particle number and finite kinetic energy, and, possibly,
with particle number not exceeding $N$.

Hartree-Fock theory is well established in non-relativistic quantum
mechanics. It is an upper bound on the true quantum energy and agrees
energetically for heavy atoms with the underlying non-relativistic
quantum mechanics up to the subleading order in the atomic number (see
Lieb and Simon \cite{LiebSimon1977} for the leading order, Siedentop
and Weikard
\cite{SiedentopWeikard1986,SiedentopWeikard1987O,SiedentopWeikard1988,SiedentopWeikard1989}
(upper and lower bound) and Hughes \cite{Hughes1986,Hughes1990} (lower
bound) for the leading correction, and Bach \cite{Bach1993} and
Fefferman and Seco
\cite{FeffermanSeco1989,FeffermanSeco1990O,FeffermanSeco1992,FeffermanSeco1993,FeffermanSeco1994}
for the subleading order). Moreover agreements of the densities are
shown on the scale $Z^{-\frac13}$ in \cite{LiebSimon1977} and on
smaller scales by Iantchenko \textit{et al.}
\cite{Iantchenkoetal1996}, Iantchenko \cite{Iantchenko1997} and by
Iantchenko et al. \cite{IantchenkoSiedentop2001} (density matrix). The
same accuracy of the energy has been shown for the M\"uller functional
\cite{Siedentop2009,Siedentop2014}. It is conjectured a lower bound on
the quantum energy. However, this was only shown for two electrons
(Frank et al. \cite{Franketal2007} based on an inequality of Wigner
and Yanase \cite{WignerYanase1964}.

Although the CA functional has been already successfully benchmarked
numerically by Jara-Cortes et al. \cite[Table 1]{Jara-Cortesetal2022},
an analytic study is lacking. Our goal is to fill this gap and to show
how the CA functional fits in the picture from an analytical point of
view. In particular we will show that it is actually sandwiched in
between the Hartree-Fock and the M\"uller functional. Establishing
this result will have the consequence that the CA ground state energy
of atoms of atomic number $Z$ agrees with the true quantum energy up
subsubleading order in $Z$.

We start with some preparatory facts.
\begin{lemma}
  Assume $\mu,\lambda\in[0,1]$. Then
  \begin{equation}
    \label{eq:lm}
    \lambda\mu+\sqrt{\lambda(1-\lambda)\mu(1-\mu)} \leq \sqrt{\lambda\mu}.  \end{equation}
\end{lemma}
\begin{proof}
  Applying the Schwarz inequality to the sum on the left side of
  \eqref{eq:lm} give
  \begin{equation}
    \lambda\mu+\sqrt{\lambda(1-\lambda)\mu(1-\mu)} \leq \sqrt{\lambda^2+\lambda(1-\lambda)}\sqrt{\mu^2+\mu(1-\mu)}=\sqrt{\lambda\mu} 
  \end{equation}
  which is identical with the right side of \eqref{eq:lm}.
  \end{proof}
We also remind of the formula (Fefferman and de la Llave
\cite[p. 124]{FeffermandelaLlave1986})
\begin{equation}
  \label{eq:FeffermandelaLlave}
  {1\over |\bx-\by|} = {1\over\pi} \int_0^\infty{\rd r\over r^5} \int_{\rz^3}\rd \bz\; B_{r,\bz}(\bx)B_{r,\bz}(\bx)
\end{equation}
where $B_{r,\bz}$ is the characteristic function of the unit ball centered around $\bz$, i.e.,
\begin{equation}
  B_{r,\bz}(\bx):=
  \begin{cases}
    1&|\bx-\bz|<r \\
    0&|\bx-\bz|\geq r
  \end{cases}.
\end{equation}

Writing
\begin{equation}
  \label{eq:M}
  E_\mathrm{M}(Z):=\inf\m(\gQ_Z)
\end{equation}
for the neutral atomic M\"uller ground state energy our basic result
is
\begin{theorem}
  \label{Energie-ca}
  Assume $\gamma\in \gQ$. Then
  \begin{equation}
    \label{eq:ungleichung}
    \m(\gamma)\leq \ca(\gamma)\leq\hf(\gamma).
  \end{equation}
\end{theorem}
\begin{proof}
    The upper bound on $\ca(\gamma)$ is immediate from the definition of $\ca$. The remaining lower bound is equivalent to showing
    \begin{equation}
      \label{eq:aequiungleichung}
      X[\gamma]+X[\sqrt{\gamma(1-\gamma)}]\leq X[\sqrt\gamma].
    \end{equation}

    First note that
    \begin{equation}
      0\leq\int_\Gamma\rd x \int_\Gamma\rd y B_{r,\bz}(x) \xi_m(x)\overline{\xi_n(x)} \overline{\xi_m(y)\overline{\xi_n(y)}} B_{r,\bz}(\by)
    \end{equation}
    and therefore, choosing $\lambda=\lambda_m$ and $\mu=\lambda_n$ in
    \eqref{eq:lm},
    \begin{multline}
      0\leq \left(\sqrt{\lambda_m\lambda_n}- \lambda_m\lambda_n -\sqrt{\lambda_m(1-\lambda_m)\lambda_n(1-\lambda_n)}\right)\\
      \times \int_\Gamma\rd x \int_\Gamma\rd y B_{r,\bz}(x)
      \xi_m(x)\overline{\xi_n(x)}
      \overline{\xi_m(y)\overline{\xi_n(y)}} B_{r,\bz}(\by).
    \end{multline}
    We divide the inequality by $2\pi r^5$ and integrate over $r$ and
    $\bz$ and use Fefferman's and de la Llave's formula
    \eqref{eq:FeffermandelaLlave} for the Coulomb kernel and get
\begin{multline}
  \tfrac12 \int_\Gamma\rd x \int_\Gamma\rd y
  {|\sum_m\lambda_m\xi_m(x)\overline{\xi_m(y)}|^2+|\sum_m\sqrt{\lambda_m(1-\lambda_m)}\xi_m(x)\overline{\xi_m(y)}|^2\over|\bx-\by|}\\
  \leq
   \tfrac12 \int_\Gamma\rd x \int_\Gamma\rd y
  {|\sum_m\sqrt{\lambda_m}\xi_m(x)\overline{\xi_m(y)}|^2\over|\bx-\by|}
\end{multline}
which is \eqref{eq:aequiungleichung} written in terms of the
eigenfunctions and eigenvalues of $\gamma$.
\end{proof}
Theorem \ref{Energie-ca} has an immediate consequence on the expansion
of ground state energy
$E_Q(Z):= \inf \sigma(H_{Z,Z})$ of the neutral atomic Hamiltonian (in the
$N$-electron space $\atp_{n=1}^NL^2(\Gamma:\cz)$)
\begin{equation}
  \label{eq:hnz}
  H_{N,Z}:= \sum_{n=1}^N\left( {\bp^2_n\over2}-{Z\over|\bx_n|}\right)
  +\sum_{1\leq m<n\leq N}{1\over|\bx_m-\bx_n|}
\end{equation}
in terms of the CA ground state energy
\begin{equation}
  \label{eq:groundstateenergy-ca}
  E_\mathrm{CA}(Z):=\inf\ca(\gQ_Z).
\end{equation}
We get
\begin{equation}
  \boxed{\label{Korrolar-ca}
 E_\mathrm{Q}(Z)= E_\mathrm{CA}(Z)+ o(Z^\frac53)=
e_\mathrm{TF}Z^\frac73+\tfrac12Z^2 - c_\mathrm{S}Z^\frac53 +
o(Z^\frac53)}
\end{equation}
where $e_\mathrm{TF}$ is the Thomas-Fermi energy of hydrogen and $c_S$
is the Schwinger constant.
\begin{proof}[Proof of Formula \eqref{Korrolar-ca}]  
  We recall the following three fundamental facts:
  \begin{description}
  \item[Quantum Ground State Energy] Fefferman and Seco
    \cite{FeffermanSeco1989,FeffermanSeco1990O,FeffermanSeco1992,FeffermanSeco1993,FeffermanSeco1994}
    showed
    \begin{equation}
      \label{eq:FS}
      E_\mathrm{Q}(Z)=e_\mathrm{TF}Z^\frac73+\tfrac12Z^2 - c_\mathrm{S}Z^\frac53 +
o(Z^\frac53).
    \end{equation}
  \item[Quantum versus HF Ground State Energy] Bach \cite{Bach1993} showed
    \begin{equation}
      \label{eq:b}
     E_\mathrm{Q}(Z)- E_\mathrm{HF}(Z)=o(Z^\frac53).
    \end{equation}
  \item[Quantum versus M\"uller] The difference of M\"uller and
    Hartree-Fock energy fulfills (see \cite[Satz 1]{Siedentop2009} and
    also \cite{Siedentop2014})
    \begin{equation}
      \label{eq:m}
       E_\mathrm{M}(Z)-E_\mathrm{HF}(Z)=o(Z^\frac53).
    \end{equation}
  \end{description}
  These three facts together with the sandwiching of the CA-energy
  proven in Theorem \ref{Energie-ca} gives \eqref{Korrolar-ca}.
  
\end{proof}


\def\cprime{$'$}

\end{document}